\documentclass[12pt,a4paper]{article}
\usepackage{jheppub}

\usepackage{color}
\usepackage{amsmath}

\usepackage{verbatim}   
\usepackage{subfigure}  
\usepackage{amsfonts}
\usepackage{amssymb}
\usepackage{mathrsfs}
\usepackage{graphicx}
\usepackage{slashed}
\usepackage{amsthm}

\def\bar#1{\overline{#1}}

\def\inv{^{\raise.15ex\hbox{${\scriptscriptstyle -}$}\kern-.05em 1}}
\def\lbar{{\lower.35ex\hbox{$\mathchar'26$}\mkern-10mu\lambda}} 

\let\p=\partial
\let\<=\langle
\let\>=\rangle

\let\+=\uparrow

\def\a{\alpha'}

\def\d{\textrm{d}}
\def\End{\textrm{End}}

\def\ker{\textrm{ker}}

\def\bp{\bar{\partial}}
\def\tr{\textrm{tr}}
\def\OO{\mathcal{O}}

\def\Q{\mathcal{Q}}

\def\H{\mathcal{H}}

\def\P{\mathcal{P}}

\def\H{\mathcal{H}}
\def\M{\mathcal{M}}

\newtheorem{Theorem}{Theorem}
\newtheorem{Conjecture}{Conjecture}
\newtheorem{Proposition}{Proposition}

\theoremstyle{definition}

\begin{document}



\title{\rm Connections, Field Redefinitions and Heterotic Supergravity\\}

\author[a]{Xenia de la Ossa}
\author[b,c,d,e]{and Eirik E.~Svanes}

\affiliation[a]{Mathematical Institute, University of Oxford, Andrew Wiles Building\\
Radcliffe Observatory Quarter, Woodstock Road, Oxford OX2 6GG, UK\\}
\affiliation[b]{Rudolf Peierls Centre for Theoretical Physics, University of Oxford\\ 
1 Keble Road, Oxford OX1 3NP, UK}
\affiliation[c]{Sorbonne Universités, UPMC Univ Paris 06, UMR 7589, LPTHE, F-75005, Paris, France}
\affiliation[d]{CNRS, UMR 7589, LPTHE, F-75005, Paris, France}
\affiliation[e]{Sorbonne Universités, Institut Lagrange de Paris, 98 bis Bd Arago, 75014 Paris, France\\}

\emailAdd{delaossa@maths.ox.ac.uk}
\emailAdd{esvanes@lpthe.jussieu.fr}

\null\vskip10pt

\abstract{We study heterotic supergravity at $\OO(\a)$, first described in detail in 1989 by Bergshoeff and de Roo. In particular, we discuss an ambiguity of a connection choice on the tangent bundle. It is well known that at $\OO(\a)$ the Hull connection gives a consistent supergravity theory with supersymmetry transformations given in the usual way. We consider deformations of this connection corresponding to field redefinitions, and the necessary corrections to the supersymmetry transformations. We are interested in the moduli space of such field redefinitions which allow for supersymmetric solutions to the equations of motion. We show that for solutions on $M_4\times X$, where $M_4$ is Minkowski and $X$ is compact, the moduli space of infinitesimal field redefinitions is given by $H^{(0,1)}(X,\End(TX))$. This space corresponds to infinitesimally close connections for which the equations of motion are satisfied. The setup suggests a symmetry between the gauge connection and the tangent bundle connection, as also employed by Bergshoeff and de Roo. We argue that this symmetry should be kept to higher orders in $\a$, and propose a natural choice for the corresponding tangent bundle connection used in curvature computations. In particular, the Hull connection should be corrected at second and higher orders in $\a$ from this point of view.}


\maketitle

\newpage

\section{Introduction}

This paper is a continuation of \cite{delaOssa:2014cia} where we studied the infinitesimal moduli space of heterotic supergravity, and in particular the Strominger system \cite{Strominger1986, Hull1986357, Lust:1986ix}. Moduli of the Strominger system where also studied in \cite{Anderson:2014xha}.

In this paper we will shed more light on an ambiguity which appeared in \cite{delaOssa:2014cia} concerning moduli of $TX$ as a holomorphic bundle defined by a holomorphic connection $\nabla$. We argued that these moduli could not be physical, and formulated a possible interpretation for their appearance, which we elaborate on in this paper. 

Ambiguities concerning the connection $\nabla$ have been discussed extensively in the literature before, both from a sigma model perspective  \cite{Hull1986187, Sen1986289, 0264-9381-4-6-027, Melnikov:2012cv, Melnikov:2012nm}, where a change of this connection has been shown to correspond to a field redefinition, or from the supergravity point of view  \cite{Hull198651, Hull1986357, Becker:2009df, Melnikov:2014ywa}, where a change of connection choice has been shown to correspond to a change of regularisation scheme in the effective action. We will review some of these results and extend them to higher orders in $\a$. Heterotic supergravity has also been considered at higher orders in $\a$ before \cite{Witten:1986kg, Bergshoeff:1988nn, Bergshoeff1989439, Gillard:2003jh, Anguelova:2010ed},\footnote{More recently, $\a$-corrections in heterotic supergravity have been considered in the context of generalised geometry and double field theory \cite{Bedoya:2014pma, Hohm:2014eba, Hohm:2014xsa, Coimbra:2014qaa}, which is more closely related to the structure presented in \cite{delaOssa:2014cia}.} and we will review and extend some of these results. In particular, we address the ambiguity concerning a connection choice on the tangent bundle in the higher order theory. 

We begin in section \ref{sec:FirstOrder} with a short review of first order heterotic supergravity as first written down in \cite{Bergshoeff:1988nn, Bergshoeff1989439}. We discuss the connection choice $\nabla$ on the tangent bundle $TX$ needed for the supersymmetry equations and equations of motion to be compatible. This leads to the an instanton condition on $\nabla$ \cite{Ivanov:2009rh, Martelli:2010jx, delaOssa:2014cia}
\begin{equation*}
R_{mn}\Gamma^{mn}\eta=0\:,
\end{equation*}
where $R_{mn}$ is the curvature two-form of $\nabla$ and $\eta$ is the spinor parametrising supersymmetry on $X$. This condition has often been applied in the literature when constructing heterotic vacua, see e.g. \cite{Harland:2011zs, Gemmer:2012pp, Chatzistavrakidis:2012qb, Klaput:2012vv, Haupt:2014ufa}. The instanton condition has an associated infinitesimal moduli space
\begin{equation}
\label{eq:Mnabla}
T\M_\nabla\cong H^{(0,1)}(X,\End(TX))\:.
\end{equation}
which we considered in \cite{delaOssa:2014cia}. These moduli cannot be physical, and the main purpose of this paper is an attempt to understand their appearance. 

In order to have a consistent supergravity at $\OO(\a)$, $\nabla$ must be taken to be the Hull connection \cite{Hull1986357, Bergshoeff1989439}. A deformation of this connection is equivalent to a field redefinition or a change of the regularisation scheme in the effective action \cite{Sen1986289, Hull198651}. We are interested in the space of allowed deformations, for which there are supersymmetric solutions to the supergravity equations of motion. We find that even though we need to deform the supersymmetry transformations accordingly, as was also pointed out in \cite{Melnikov:2014ywa}, the conditions for preservation of supersymmetry can be assumed to remain the same. Moreover, the space of connections which allow for such supersymmetric solutions to exist is again given by \eqref{eq:Mnabla}.  

In section \ref{sec:higher} we discuss extensions of these results to second order in $\a$. We find that the choice of the Hull connection, which was required at at $\OO(\a)$, is no longer consistent at higher orders. Indeed, as we shall see, insisting on the Hull connection can put additional constraints on the higher order geometry. In particular, the first order geometry is Calabi-Yau if we assume the internal space to be compact and smooth. This was also noted in \cite{Ivanov:2009rh}, where the first order geometry was taken as exact.

We also show that without loss of generality supersymmetric solutions may be assumed to satisfy the Strominger system, assuming that the internal space is compact and smooth. With this, the connection $\nabla$ should again satisfy the instanton condition. This condition looks surprisingly like a supersymmetry condition corresponding to the connection $\nabla$ as if it was a dynamical field. Indeed, it was precisely the fact that $(\nabla,\psi_{IJ})$, where $\psi_{IJ}$ is the supercovariant curvature, transforms as an $SO(9,1)$-Yang-Mills multiplet at $\OO(\a)$ which led to the construction of the $\OO(\a)$-action in the first place \cite{Bergshoeff1989439}. As also noted in \cite{Bergshoeff1989439}, this is symmetric with the gauge sector of the theory, and it is natural to assume this symmetry to higher orders in $\a$. This also prompts us to make a conjecture for what the connection choice should be at higher orders in $\a$.

We have left some technical details of the discussion of compactifications to four-dimensional Minkowski space to the Appendix, leaving us free to discuss supersymmetry and solutions in the bulk of the paper.

\section{First Order Heterotic Supergravity}
\label{sec:FirstOrder}
In this section we review heterotic supergravity at first order in $\a$ as first studied in \cite{Bergshoeff:1988nn, Bergshoeff1989439}. We write down the action and supersymmetry transformations, and review the supersymmetric solutions of this theory, commonly known as the Strominger sytem \cite{Strominger1986, Hull1986357, Lust:1986ix}. We show how consistency between the supersymmetry conditions and equations of motion requires that we make a certain connection choice on $TX$. In fact, the connection should satisfy the instanton condition. Various proofs of this have appeared in the literature before \cite{Ivanov:2009rh, Martelli:2010jx, delaOssa:2014cia}, and we give a slightly different proof in this paper. 

Moreover, the condition of a supersymmetry invariant supergravity action reduces this choice of connection further to the the Hull connection \cite{Hull1986357, Bergshoeff:1988nn, Bergshoeff1989439}. We show that by deforming this connection, the supersymmetry transformations must be transformed accordingly, however we find that the conditions for supersymmetric solutions can be taken as before. Moreover, the deformed connections must again be instantons. We comment on the moduli space of this condition, and note that these moduli are unphysical as they correspond to changes of the effective action regularisation scheme \cite{Hull198651}, or as we shall see explicitly, field redefinitions \cite{Sen1986289}. 

We comment briefly on the type of geometry that results from the first order supersymmetry conditions, and the fact that the compact space $X$ is conformally balanced. First order deformations of the corresponding system of equations where studied recently in \cite{delaOssa:2014cia} and \cite{Anderson:2014xha}.

\subsection{Action and Field Content}
Let us begin by recalling the bosonic part of the action at this order \cite{Bergshoeff1989439}\footnote{Although this action is valid to second order in $\a$, we only need it to first order in this section.}
\begin{align}
S=&
\int_{M_{10}}e^{-2\phi}\Big[*\mathcal{R}-4\vert\d\phi\vert^2+\frac{1}{2}\vert H\vert^2+\frac{\alpha'}{4}(\tr\:\vert F\vert^2-\tr\:\vert R\vert^2)\Big]+\OO(\a^3)\:,
\label{eq:action}
\end{align}
where $\mathcal{R}$ is the scalar curvature of the metric $g$, $F$ is the curvature of an $E_8\times E_8$ gauge bundle, $R$ is the curvature with respect to some connection $\nabla$ on the tangent bundle, and $H$ is the NS three-form,
\begin{equation}
\label{eq:anomalycancellation}
H=\d B+\frac{\a}{4}(\omega_{CS}^A-\omega^\nabla_{CS})\:,
\end{equation}
which is appropriately defined for the theory to be anomaly free. Here, $\omega^A_{CS}$ and $\omega^\nabla_{CS}$ are Chern-Simons three-forms of the gauge-connection $A$, and the tangent bundle connection $\nabla$ respectively. We also write $\vert\alpha\vert^2=\alpha\wedge*\alpha$ for $\alpha\in\Omega^*(X)$.

The fermonic fields of the theory are the gravitiono $\psi_M$, the dilatino $\lambda$ and the gaugino $\chi$. The $N=1$ supersymmetry variations of these fields are \cite{Strominger1986, Bergshoeff1989439},
\begin{align}
\label{eq:O1spinorsusy1}
\delta\psi_M &=\nabla^+_M\epsilon=\Big(\nabla^{\hbox{\tiny{LC}}}_M+\frac{1}{8}\H_M\Big)\epsilon+\OO(\a^2)\:,\\
\label{eq:O1spinorsusy2}
\delta\lambda &=\Big(\slashed\nabla^{\hbox{\tiny{LC}}}\phi+\frac{1}{12}\H\Big)\epsilon+\OO(\a^2)\:,\\
\label{eq:O1spinorsusy3}
\delta\chi &=-\frac{1}{2}F_{MN}\Gamma^{MN}\epsilon+\OO(\a)\:,
\end{align}
where $\H_M=H_{MNP}\Gamma^{NP}$, $\H=H_{MNP}\Gamma^{MNP}$, $\nabla^{\hbox{\tiny{LC}}}$ denotes the Levi-Civita connection, the $\Gamma^{M}$ are ten-dimensional gamma-matrices, and $\Gamma^{M_1..M_n}$ denote antisymmetrized products of gamma-matrices as usual. Here large roman letters denote ten-dimensional indices. Note that the transformation for the gauge field has a reduction in the order of $\a$. This is because the gauge field always appears with an extra factor of $\a$ in the action. In order to have a supersymmetric theory, we therefore only need to specify the gaugino transformation modulo $\OO(\a)$-terms. Supersymmetry for a given solution then requires that these variations are set to zero.

The choice of connection $\nabla$ is a subtle question. Firstly it cannot be a dynamical field, as there are no modes in the corresponding string theory corresponding to this. Hence, $\nabla$ must depend on the other fields of the theory in a particular way. This dependence is forced upon us once the supergravity action and supersymmetry transformations are specified.

Indeed, if we want the supergravity action to be invariant under the supersymmetry transformations \eqref{eq:O1spinorsusy1}-\eqref{eq:O1spinorsusy3} at $\OO(\a)$, we need a particular choice of connection in the action, namely the Hull connection $\nabla^-$, which has connection symbols given by
\begin{equation}
\label{eq:connsymb}
{\Gamma^{\pm}_{KL}}^M={\Gamma^{\hbox{\tiny{LC}}}_{KL}}^M\pm\frac{1}{2}{H_{KL}}^M\:,
\end{equation}
where ${\Gamma^{\hbox{\tiny{LC}}}_{KL}}^M$ are the connection symbols of the Levi-Civita connection. This connection is needed in order that $(\nabla,\psi_{IJ})$ transforms as a $SO(9,1)$ Yang-Mills multiplet, as explained in \cite{Bergshoeff1989439}. Here $\psi_{IJ}$ is the supercovariant curvature given by
\begin{equation}
\label{eq:scurv}
\psi_{IJ}=\nabla^+_I\psi_J-\nabla^+_J\psi_I\:,
\end{equation}
where $\nabla^+$ has connection symbols given by \eqref{eq:connsymb}. With this, the full first order heterotic action is invariant under supersymmetry.

\subsection{First Order Supersymmetry and Geometry}
In this section we briefly consider compactifications to four-dimensional Minkowski space. Details of the compactifications are laid out in Appendix \ref{app:comp}, where conventions are laid out. Consider the  set of supersymmetry equations \eqref{eq:O1spinorsusy1}--\eqref{eq:O1spinorsusy2} at first order in $\a$. We look at what conditions they impose on the internal geometry $X$. The resulting system is known as the Strominger system \cite{Strominger1986, Hull1986357}, and in terms of the fields  $(\Psi, \omega, H, \phi)$ they may be written as \cite{Gauntlett:2001ur}
\begin{align}
\label{eq:susy1}
\d(e^{-2\phi}\:\Psi)=&0\:,\\
\label{eq:susy2}
\d(e^{-2\phi}\:\omega\wedge\omega)=&0\:,\\
\label{eq:susy3}
-e^{2\phi}\:\d(e^{2\phi}\:\omega)=&*H\:,
\end{align}
Equation \eqref{eq:susy1} implies the existence of a holomorphic three-form $\Omega=e^{-2\phi}\Psi$. It also implies that the complex structure $J$ defined by $\Psi$ (as in Appendix \ref{app:comp}) is integrable, which means that $\bar\p^2=0$, where $\p$ and $\bar\p$ are the Dolbeault operators defined by the complex structure $J$. A complex three-fold $X$ satisfying equation \eqref{eq:susy2} is said to be conformally balanced. 

In terms of the Dolbeault operators, it may be shown that the flux may be written as 
\begin{equation}
\label{eq:flux2}
H=i(\p-\bar\p)\omega=T\:,
\end{equation}
as was first shown in \cite{Strominger1986}. The flux $H$ is identified with the torsion $T$ of the Bismut connection $\hat\nabla$, which is in fact the same as $\nabla^+$. In the mathematics literature, $\hat\nabla$ is the unique metric connection with totally antisymmetric torsion so that
\begin{equation*}
\hat\nabla J=0\:.
\end{equation*}

From the anomaly cancellation condition \eqref{eq:anomalycancellation}, we also have the Bianchi Identity
\begin{equation}
\d H=\frac{\a}{4}(\tr\: F\wedge F-\tr\: R\wedge R)\:.
\end{equation}
Setting the gaugino variation \eqref{eq:O1spinorsusy3} to zero is equivalent to requiring that the gauge-bundle is holomorphic, and satisfies the hermitian Yang-Mills equations on the internal space
\begin{equation}
\label{eq:inst1}
F\wedge\Psi=0,\;\;\;\; F\wedge\omega\wedge\omega=0,
\end{equation}
where $F$ is the field-strength of the $E_8\times E_8$ gauge-bundle.

\subsection{Instanton Condition}
\label{sec:instcond}

In order for the $\OO(\a)$-action to be invariant under the supersymmetry trasformations \eqref{eq:O1spinorsusy1}-\eqref{eq:O1spinorsusy3}, one is forced to choose the Hull connection in the action. As we shall see in section \ref{sec:supsol1}, this can be relaxed upon appropriate field redefinitions. Such field redefinitions also change the supersymmetry transformations. We will however see that even though the supersymmetry transformtions change, the {\it supersymmetry conditions} may be assumed to be the same. That is, we can without loss of generality, assume that our solutions solve the Strominger system. Furthermore, supersymmetry should be compatible with the bosonic equations of motion derived from \eqref{eq:action}. This leads to a condition on $\nabla$ known as the instanton condition which we now discuss. It should be noted that for supersymmetric solutions, as we also show in Appendix \ref{app:Hull}, the Hull connection does satisfy the instanton condition to the order we are working at \cite{Martelli:2010jx}.

It has been shown that the supersymmetry equations derived from \eqref{eq:O1spinorsusy1}-\eqref{eq:O1spinorsusy3} together with the Bianchi identity imply the equations of motion if and only if the connection $\nabla$ for the curvature two-form $R$ appearing in \eqref{eq:action} is an $SU(3)$-instanton \cite{Ivanov:2009rh, Martelli:2010jx}. This implies that it satisfies the conditions
\begin{equation}
\label{eq:inst2}
R\wedge\Omega=0,\;\;\;\;R\wedge\omega\wedge\omega=0\:,
\end{equation}
which are similar to those for the field-strength $F$. We present a proof of this in Appendix \ref{app:proof} for completeness.\footnote{A more general proof, similar to the one found in Appendix \ref{app:proof}, can also be found in \cite{Held:2010az}.} The first condition in \eqref{eq:inst2} implies that $R^{(0,2)}=0$. Therefore there is a holomorphic structure $\bp_\vartheta$ on the tangent bundle $TX$, where
\begin{equation*}
\bp_\vartheta=\bp+[\vartheta,\:]\:,
\end{equation*}
and $\vartheta$ is the $(0,1)$-part of the connection one-form of $\nabla$. We denote $TX$ with this holomorphic structure as $(TX,\nabla)$. Note that this holomorphic structure is different in general from the holomorphic structure on $TX$ induced by the complex structure $J$.

The second condition of \eqref{eq:inst2} says that the connection $\nabla$ is Yang-Mills, more precisely, $\nabla$ is an instanton. By a theorem of Li and Yau \cite{Yau87}, which generalizes the Donaldson Uhlenbeck-Yau theorem \cite{donaldson1985anti, uhlenbeck1986existence}, such a connection exists if and only if the holomorphic bundle $(TX,\nabla)$ is poly-stable. Moreover, the connection is the unique hermitian connection with respect to the corresponding hermitian structure on $TX$.\footnote{This theorem holds on complex manifolds with a Gauduchon metric. In our case, the Gauduchon metric is given by $e^{-\phi}\:\omega$.}

It is known that the stability condition is stable under first order deformations of the holomorphic structure \cite{MR2665168}. We extended the result of \cite{MR2665168} in \cite{delaOssa:2014cia}, where we found that the moduli space of infinitesimal deformations of $\bp_\vartheta$, including generic deformations of the Hermitian Yang-Mills conditions \eqref{eq:inst2}, where given by
\begin{equation}
\label{eq:modulitheta}
T\M_{\bp_\vartheta}=H^{(0,1)}_{\bp_\vartheta}(X,\End(TX))\:.
\end{equation}
More explicitly, in \cite{delaOssa:2014cia} we showed that for each $[\delta\vartheta]\in T\M_{\bp_\vartheta}$,\footnote{Here $[\delta\vartheta]$ denote equivalence classes of deformations modulo gauge transformations.} there is a corresponding element $\delta\vartheta\in[\delta\vartheta]$ so that the Yang-Mills condition is satisfied. Starting from the instanton connection, there is then an infinitesimal moduli space $T\M_{\bp_\vartheta}$ of connections for which the equations of motion are satisfied.

As mentioned, for the supergravity action to be invariant under the supersymmetry transformations \eqref{eq:O1spinorsusy1}-\eqref{eq:O1spinorsusy3}, the choice of connection is reduced further. In particular, invariance of the first order action forces the connection to be the Hull connection $\nabla^-$ \cite{Bergshoeff1989439}. Under these supersymmetry transformations, we therefore cannot choose any element in $T\M_{\bp_\vartheta}$ when deforming the Strominger system. Rather we have to choose the element corresponding to a deformation of the Hull connection.

\subsection{Changing the Connection}
We could ask what happens if we deform the connection in the action. Firstly, such deformations do not correspond to physical fields. We shall see in this section that they are equivalent to field redefinitions \cite{Sen1986289}. Secondly, insisting upon changing this connection means that we need to change the supersymmetry transformations correspondingly. However, it turns out that the conditions for supersymmetric solutions can be taken as before. Moreover, the condition that the new connection allows for such supersymmetric solutions to the theory forces the new connection precisely to satisfy the instanton condition. 

Let us discuss what happens when we change the connection $\nabla$ used in the action. That is, we let
\begin{equation}
\label{eq:connchange}
\nabla=\nabla^-+t\theta\:,
\end{equation}
where $\theta=\theta(\Phi)$ is a function of all the other fields of the theory, which we collectively have denoted by $\Phi$, and $t$ is an initesimal parameter. In the next section, we will take $t=\OO(\a)$, but for now we just assume it corresponds to an infinitesimal deformation of the connection. We are interested in what happens to the theory under such a small deformation.

Under supersymmetry, the new connection one-forms ${\Theta_I}^{JK}$ together with the supercovariant curvature $\psi_{IJ}$ transform as
\begin{align*}
\delta{\Theta_I}^{JK}&={(\delta\Theta^-+t\delta\theta)_I}^{JK}=\frac{1}{2}\bar\epsilon\:\Gamma_I\psi^{JK}+t{\delta\theta_I}^{JK}+\OO(\a)\:,\\
\delta\psi_{IJ}&=-\frac{1}{4}R^+_{IJKL}\Gamma^{KL}\:\epsilon=-\frac{1}{4}R^-_{KLIJ}\Gamma^{KL}\:\epsilon+\OO(\a)\\
&=-\frac{1}{4}\Big(R_{KLIJ}-t(\d_{\nabla^-}\theta)_{KLIJ}\Big)\Gamma^{KL}\:\epsilon+\OO(\a)\:,
\end{align*}
where we have used \eqref{eq:curvatureid} in the second equality of the second equation. The $\OO(\a)$-terms can be neglected to the order we are working at, but they will become important in the next section when we discuss the theory to higher orders in $\a$. We thus see that $(\Theta,\psi_{IJ})$ transforms as an $SO(9,1)$-Yang-Mills multiplet, modulo $\OO(t)$ and $\OO(\a)$-terms. As noted, the $\OO(\a)$-terms can be ignored for now, but the $\OO(t)$-terms will have to be dealt with. This is done by changing the supersymmetry transformations accordingly as we shall see below.\footnote{That a change of the connection requires a change of the supersymmetry transformations in order to have a supersymmetry invariant action has been noted before \cite{Melnikov:2014ywa}.}

A lemma of Bergshoeff and de Roo \cite{Bergshoeff1989439} (see also \cite{Andriot:2011iw}) states that the action deforms as
\begin{equation}
\label{eq:lemmabos}
\frac{\delta S}{\delta\nabla^-}\propto\a\:\textrm{B}_0+\OO(\a^2)\:,
\end{equation}
under an infinitesimal deformation of the Hull connection. Here $\textrm{B}_0$ denotes a combination of zeroth order bosonic equations of motion. As the correction to the action due to the change of connection \eqref{eq:connchange} is proportional to the equations of motion, the change of connection $t\theta$ may equivalently be viewed as an infinitesimal field redefinition of order $\OO(t,\a)$, and is therefore non-physical.\footnote{That deformations of the connection corresponds to a field redefinition has been noted in the literature before, see e.g. \cite{Hull1986187, Sen1986289, Becker:2009df, Melnikov:2014ywa}.} It should be noted that the right hand side of \eqref{eq:lemmabos} generically also contains extra fermionic terms. The change of the connection \eqref{eq:connchange} hence requires a bosonic field redefinition, in addition to a change of the fermionic sector of the action. Appropriate changes to the fermionic sector are however rather common, and where utilized in e.g. \cite{Bergshoeff1989439} to write the supersymmetry transformations in a convenient way. See footnote \ref{foot:note} below.

We want to consider what happens to the theory under these deformations of the connection. In particular, we are interested in the allowed deformations of the connection, or equivalently field redefinitions, for which supersymmetric solutions to the Strominger system exist. We expect this to be related to the moduli space of connections \eqref{eq:modulitheta} studied in \cite{delaOssa:2014cia}, and we see that this is indeed the case.

From \eqref{eq:lemmabos} it follows that the change to the action due to the correction,
\begin{equation*}
{\delta_t(\delta\Theta)_I}^{JK}=t{\delta\theta_I}^{JK}\:,
\end{equation*}
of the transformation of $\Theta$ can be absorbed in a redefinition of the bosonic supersymmetry transformations by a similar procedure as is done in \cite{Bergshoeff1989439} for the $\OO(\a^2)$-corrections to the supersymmetry transformations. Similarly, we also have
\begin{equation}
\label{eq:lemmaferm}
\frac{\delta S}{\delta\psi_{IJ}}\propto\a\Psi^{IJ}_0+\OO(\a^2)\:,
\end{equation}
by \cite{Bergshoeff1989439}, where $\Psi_0$ is a combination of zeroth order fermionic and bosonic equations of motion. It follows that the change in the action due to the correction,
\begin{equation*}
\delta_t(\delta\psi)_{IJ}=\frac{t}{4}(\d_{\nabla^-}\theta)_{KLIJ}\Gamma^{KL}\epsilon\:,
\end{equation*}
may be absorbed into a redefinition of the fermionic supersymmetry transformations which now read
\begin{align}
\delta\psi_M=&\left(\nabla^+_M+\frac{t}{4}\:\mathcal{C}_M\right)\epsilon+\OO(\a^2)\notag\\
\label{eq:spinorsusy1def}
=&\left(\nabla^{\hbox{\tiny{LC}}}_M+\frac{1}{8}(\H_M+2t\mathcal{C}_M)\right)\epsilon+\OO(\a^2)\:,\\
\label{eq:spinorsusy2def}
\delta\lambda=&-\frac{1}{2\sqrt{2}}\left(\slashed\nabla^{\hbox{\tiny{LC}}}\phi+\frac{1}{12}(\H+3t\mathcal{C})\right)\epsilon+\OO(\a^2)\:,\\
\label{eq:spinorsusy3def}
\delta\chi=&-\frac{1}{2}F_{MN}\Gamma^{MN}\:\epsilon+\OO(\a)\:,
\end{align}
where
\begin{equation}
\label{eq:C}
C_{MAB}=\a12\:e^{2\phi}\:\nabla^{+L}e^{-2\phi}\Big((\d_{\nabla^-}\theta)_{ABLM}\Big)\:.
\end{equation}
Here $\mathcal{C}_M=C_{MAB}\Gamma^{AB}$ and $\mathcal{C}=C_{MAB}\Gamma^{MAB}$. The index labels $A,B,..$ denote tangent space (flat) indices, that is, $\Gamma^M=e^M_A\Gamma^A$, where $\{e^M_A\}$ is a ten-dimensional viel-bein frame. We have written the corrections in this way to be able to compare with the higher order $\a$-corrections in the next section. With the new supersymmetry transformations \eqref{eq:spinorsusy1def}-\eqref{eq:spinorsusy3def}, the action with the new connection is again invariant. 

As we saw above, deforming the connection $\nabla^-\rightarrow\nabla^-+t\theta$, really just corresponds to an $\OO(\a)$-field redefinition. Hence, the supersymmetry algebra above (including the bosonic transformations, which we did not write down for brevity) should just be the old algebra written in terms of the new fields. There are therefore no issues concerning closedness of the algebra.

\subsection{Supersymmetric Solutions}
\label{sec:supsol1}
Let us look for four-dimensional supersymmetric maximally symmetric compact solutions to the $t$-adjusted theory. This ammounts to setting the transformations \eqref{eq:spinorsusy1def}-\eqref{eq:spinorsusy3def} to zero. We consider solutions such that
\begin{equation}
\label{eq:Stromtcorr}
\nabla^+_m\,\eta=\OO(\a^2)\:,
\end{equation}
where $\eta$ is the six-dimensional spinor on $X$. Given the redefined supersymmetry transformations, this might seem like a restriction of allowed supersymmetric solutions. However, this is not the case, at least for compact solutions. Indeed, we have the following proposition.

\begin{Proposition}
\label{prop:Strom}
Consider heterotic compactifications to four dimensions on a smooth compact space $X$ at $\OO(\a^{2n-1})$ or less. If
\begin{equation*}
\nabla^+\eta=\OO(\a^n)\:,
\end{equation*}
then without loss of generality we may assume that
\begin{equation*}
\nabla^+\eta=0\:,
\end{equation*}
i.e. the solutions are solutions of the Strominger system.
\end{Proposition}
\begin{proof}
First note that since $\nabla^+_m\:\eta=\OO(\a^n)$, it follows that
\begin{align*}
\OO(\a^{2n})=(\nabla^+_m\:\eta,\nabla^{+m}\:\eta)=\int_X(\nabla^+_m\:\eta)^\dagger\nabla^{+m}\:\eta=\int_X\eta^\dagger\Delta_+\eta=(\eta,\Delta_+\eta)\:,
\end{align*}
upon an integration by parts.\footnote{As the Bismut connection has antisymmetric torsion, it follows that $\nabla^+_mv^m=\nabla^{\hbox{\tiny{LC}}}_mv^m$ for some vector field $v^i$. This allows the integration by parts to be carried out.} Here 
\begin{equation*}
\Delta_+\eta=-\nabla_m^+\nabla^{+m}\:\eta\:,
\end{equation*}
and $\Delta_+$ is the Laplacian of the Bismut connection.

Next expand $\eta$ in eigen-modes of $\Delta_+$,
\begin{equation*}
\vert\eta\rangle=\sum_n\alpha_n\vert\psi_n\rangle\:,
\end{equation*}
where $\{\vert\psi_n\rangle\}$ is an orthonormal basis of eigenspinors of $\Delta_+$ with corresponding eigenvalues ${\lambda_n}$, and where we have gone to bra-ket notation for convenience. We can then compute
\begin{equation*}
(\eta,\Delta_+\eta)=\langle\eta\vert\Delta_+\vert\eta\rangle=\sum_n\lambda_n\vert\alpha_n\vert^2=\OO(\a^{2n})\:.
\end{equation*}
Note that $\lambda_n\ge0$ as $\Delta_+$ is positive semi-definite. From this it follows that each term in the above sum is of $\OO(\a^{2n})$. That is  
\begin{equation}
\label{eq:sumtermO4}
\lambda_n\vert\alpha_n\vert^2=\OO(\a^{2n})\:,\;\;\;\;\forall\;\;\;\;n\:.
\end{equation}
Moreover, we know that $\vert\eta\rangle=\OO(1)$, which implies
\begin{equation*}
\vert\vert\eta\vert\vert^2=\sum_n\vert\alpha_n\vert^2=\OO(1)\:.
\end{equation*}
It follows that at least one $\alpha_k=\OO(1)$.\footnote{We treat $\a$ as a formal expansion parameter, so sums of terms of higher orders in $\a$ cannot decrease the order in $\a$.} Then, from \eqref{eq:sumtermO4}, the corresponding eigenvalue is $\lambda_k=\OO(\a^{2n})$. At the given order in $\a$, $\OO(\a^{2n-1})$, we may without loss of generality set $\lambda_k=0$. It follows that there is a spinor in the kernel of $\nabla^+$, which we may take to be $\eta$.
\end{proof}
Using Proposition \ref{prop:Strom} we may set $n=1$ to get \eqref{eq:Stromtcorr}. It also follows from equation \eqref{eq:spinorsusy1def} that we need
\begin{equation}
\label{eq:vanishC}
\mathcal{C}_m\:\eta=\OO(\a^2)\:,
\end{equation}
for the solution to be supersymmetric. From Appendix \ref{app:proof} it then follows that the corrected connection $\nabla=\nabla^-+t\theta$ should be an instanton.

It is easy to see that \eqref{eq:vanishC} is satisfied, once we know that we are working with supersymmetric solutions of the Strominger system. Plugging the connection $\nabla$ into the instanton condition, and using that $\nabla^-$ is an instanton at this order, we find
\begin{equation}
(\d_{\nabla^-}\theta)_{mn}\Gamma^{mn}\:\eta=\OO(\a)\:,
\end{equation}
precicely the condition for the deformed connection to remain an instanton. From this, it also follows that
\begin{align*}
\mathcal{C}_m\eta&=\a12e^{2\phi}\nabla^{+n}e^{-2\phi}\Big((\d_{\nabla^-}\theta)_{ABnm}\Big)\Gamma^{AB}\eta\\
&=\a12e^{2\phi}\nabla^{+n}e^{-2\phi}\Big((\d_{\nabla^-}\theta)_{ABnm}\Gamma^{AB}\eta\Big)\\
&=\OO(\a^2)\:,
\end{align*}
as desired.
 
Finally, we remark that as noted in \cite{delaOssa:2014cia} there is an infinitesimal moduli space 
\begin{equation}
\label{eq:Mtheta}
T\M_{\bp_{\vartheta^-}}=H^{(0,1)}_{\bp_{\vartheta^-}}(X,\End(TX))
\end{equation}
of connections satisfying this condition, where the tangent space $T\M_{\bp_{\vartheta^-}}$ is taken {\it at the Hull connection}. Each connection in this moduli space corresponds to a field redefinition of the supergravity with the corresponding change of the supersymmetry transformations, \eqref{eq:spinorsusy1def}-\eqref{eq:spinorsusy3def}. Compact supersymmetric solutions of these equations may by Proposition \ref{prop:Strom} be assumed to be solutions of the Strominger system, and they also solve the equations of motion provided $\theta\in T\M_{\bp_{\vartheta^-}}$. From this perspective, the moduli space \eqref{eq:Mtheta} found in \cite{delaOssa:2014cia} is unphysical. That is, the moduli space \eqref{eq:Mtheta} may then be viewed as the space of allowed infinitesimal $\OO(\a)$ field redefinitions for which the equations of motion and supersymmetry are compatible.\footnote{Preserving the Bianchi identity puts additional requirements on $\theta$. These where worked out in detail in \cite{delaOssa:2014cia}. We neglect these issues in this paper as they are beside the point we wish to make here.}

\section{Higher Order Heterotic Supergravity}
\label{sec:higher}
Having discussed the first order theory, we now consider heterotic supergravity at higher orders in $\alpha'$. We continue our investigation from a ten-dimensional supergravity point of view, by a similar analysis as that of Bergshoeff and de Roo \cite{Bergshoeff1989439}. In \cite{Bergshoeff1989439} the Hull connection was used at higher orders as well. We wish to generalize this analysis a bit, and allow for a more general connection choice in the action, as was done in the previous section. In order not to overcomplicate matters unnecessarily, we return to letting the $TX$-connection be the Hull connection at $\OO(\a)$, which is needed in order that the full action be invariant under the usual supersymmetry transformations \eqref{eq:O1spinorsusy1}-\eqref{eq:O1spinorsusy3} at $\OO(\a)$. We will however allow this connection to receive corrections at $\OO(\a^2)$. 

There are two important points which we wish to emphasise in this section. Firstly, as we saw in the last section, we may deform the tangent bundle connection away from the Hull connection provided we deform the supersymmetry transformations correspondingly. We take a similar approach in this section, where we deform away from the Hull connection by an $\a$-correction, $\nabla=\nabla^-+\theta$ where now $\theta=\OO(\a)$, and depends on the fields of the theory. Our findings from the previous section also persist in this section. That is, the deformation $\theta$ now corresponds to an $\OO(\a^2)$ field redefinition, and deforming $\theta$ is therefore non-physical in this sense. Moreover, the supersymmetry transformations also change with $\theta$, in accordance with the deformed fields. However, not all field choices allow for supersymmetric solutions of the Strominger system.

Secondly, we note there is a symmetry between the tangent bundle connection $\nabla$ and gauge connection $A$ in the first order action. As a guiding principle, as is also done in \cite{Bergshoeff1989439}, we would like to keep this symmetry to higher orders. With this philosophy it seems natural to choose $\nabla$ so it satisfies its own equation of motion similar to that of $A$, at the locus where the equations of motion are satisfied. Note that this is true for the Hull connection at $\OO(\a)$, by equation \eqref{eq:lemmabos}.  

Moreover, this also seems to be the connection choice we need in order for the supersymmetry conditions to hold at the locus of equations of motion. Indeed, we find the following
\begin{Theorem}
\label{tm:inst}
Strominger system type solutions, where $\nabla^+\epsilon=0$ for heterotic compactifications on a compact six-dimensional manifold $X$, survive as solutions of heterotic supergravity at $\OO(\a^2)$ if and only if the connection $\nabla$ is an instanton, satisfying it's own ``supersymmetry condition"
\begin{equation}
R_{mn}\Gamma^{mn}\:\eta=0\:.
\end{equation}
Compact $\OO(\a^2)$-supersymmetric solutions can without loss of generality be assumed to be of this type. Moreover, $\nabla$ satisfies it's own equation of motion for these solutions.
\end{Theorem}

Note then that our choice of connection is as if the connection $\nabla$ was dynamical. We again stress that this is not the case, as $\nabla$ must depend on the other fields of the theory. We only require the connection to satisfy an equation of motion (as if it was dynamical), and this then relates to \it how\rm\;$\nabla$ depends on the other fields. 

With these observations, we make the following conjecture
\begin{Conjecture}
\label{con:higherorder}
At higher orders in $\a$, the correct connection choice/field choice is the choice which preserves the symmetry between $\nabla$ and $A$. That is, $\nabla$ should be chosen as if it was dynamical, satisfying it's own equation of motion. Moreover, for supersymmetric solutions, $\nabla$ should be chosen to satisfy it's own supersymmetry condition, similar to the one satisfied by $A$.
\end{Conjecture}

\subsection{The Second Order Theory}
According to Bergshoeff and de Roo \cite{Bergshoeff1989439} the bosonic part of the heterotic action does not receive corrections at this order. Bergshoeff and de Roo used the Hull connection when writing down the action, but we shall be more generic, choosing a connection $\nabla$ that differs from the Hull connection by changes of $\OO(\a)$. The bosonic action then reads
\begin{align}
S=&\int_{M_{10}}e^{-2\phi}\Big[*\mathcal{R}-4\vert\d\phi\vert^2+\frac{1}{2}\vert H\vert^2+\frac{\alpha'}{4}(\tr\:\vert F\vert^2-\tr\:\vert R\vert^2)\Big]+\OO(\a^3)\:,
\label{eq:2action}
\end{align}
where now
\begin{equation}
H=\d B+\frac{\a}{4}\left(\omega_{CS}^A-\omega^\nabla_{CS}\right)+\OO(\a^3)\:,
\end{equation}
and $\nabla=\nabla^-+\OO(\a)$.

The supersymmetry transformations do receive corrections. What these corrections are again depend crucially on which connection is chosen in the action as we will discuss in the next section. Using the Hull connection $\nabla=\nabla^-$, they are given in \cite{Bergshoeff1989439} and read\footnote{\label{foot:note}It should be noted that the specific form of these corrections, where there are no covariant derivatives of the spinor in the $\OO(\a^2)$-correction requires an addition of an extra term of $\OO(\a^2)$ to the fermionic sector action \cite{Bergshoeff1989439}.}
\begin{align}
\delta\psi_M=&\left(\nabla^+_M+\frac{1}{4}\P_M\right)\epsilon\notag\\
\label{eq:spinorsusy1}
=&\left(\nabla^{\hbox{\tiny{LC}}}_M+\frac{1}{8}(\H_M+2\P_M)\right)\epsilon+\OO(\a^3)\\
\label{eq:spinorsusy2}
\delta\lambda=&-\frac{1}{2\sqrt{2}}\left(\slashed\nabla^{\hbox{\tiny{LC}}}\phi+\frac{1}{12}\H+\frac{3}{12}\P\right)\epsilon+\OO(\a^3)\\
\label{eq:spinorsusy3}
\delta\chi=&-\frac{1}{2}F_{MN}\Gamma^{MN}\:\epsilon+\OO(\a^2),
\end{align}
where
\begin{equation}
P_{MAB}=-6\a e^{2\phi}\:\nabla^{+L}(e^{-2\phi}\d H_{LMAB})\:.
\end{equation}
Here $\P_M=P_{MAB}\Gamma^{AB}$ and $\P=P_{MAB}\Gamma^{MAB}$. Here $A, B,\:..$ denote flat indices, while $I, J,\: ..$ denote space-time indices. Note again the reduction in $\a$ for the gauge-field transformation \eqref{eq:spinorsusy3}.

\subsection{Second Order Equations of Motion}
We now derive the equations of motion of the action \eqref{eq:2action}. As the action is the same as the first order action, one might guess that the equations of motion will be the same too. This is not quite correct, and we take a moment to explain why.

When deriving the first order equations of motion, one relies on the lemma of \cite{Bergshoeff1989439}, equation \eqref{eq:lemmabos}, from which it follows that the Hull connection satisfies an equation of motion of its own, whenever the other fields do. As a necessary condition to satisfying the first order equations of motion is that the zeroth order equations of motion are of $\OO(\a)$, the variation of the action with respect to $\nabla^-$ can be ignored as it is of $\OO(\a^2)$. This simplifies matters when deriving the first order equations of motion. At second order however, such terms will have to be included, potentially leading to a more complicated set of equations.

We note that the $\OO(\a^2)$-corrections to the equations of motion come from the variation of the action with respect to $\nabla$. What they are will crutially depend on what connection $\nabla$ is used. Let us write the connection one-form of $\nabla$ as
\begin{equation*}
\Theta=\Theta^-+\theta,
\end{equation*}
where $\Theta^-$ are the connection one-forms of the Hull connection, and $\theta=\OO(\a)$, and depends on the other fields of the theory in some unspecified way. The action then takes the form
\begin{equation}
\label{eq:thetaaction}
S=S[\nabla^-]+\delta_\theta S+\OO(\a^3)\:.
\end{equation}
Let us compute $\delta_\theta S$. We find
\begin{equation*}
\delta_\theta S=\int_{M_{10}}e^{-2\phi}\Big[\delta_\theta H\wedge*H+\frac{\a}{2}\tr\:[\d_{\nabla^-}\theta\wedge*R^-]\Big]\:,
\end{equation*}
where $\d_{\nabla^-}=\d+[\Theta^-,\:]$. Now
\begin{equation*}
\delta_\theta H=\d\delta_\theta B-\frac{\a}{4}\delta_\theta\omega_{CS}^\nabla=\d\delta_\theta B+\frac{\a}{2}\tr\:[\theta\wedge R^-]+\frac{\a}{4}\d\tr\:[\theta\wedge\Theta^-]\:,
\end{equation*}
where $\delta_\theta B$ comes from the non-trivial deformation of $B$ due to the Green-Schwarz mechanism. Note that the Green-Schwarz mechanism is an $\OO(\a)$-effect, which means that $\delta_\theta B=\OO(\a^2)$. Inserting this back into the action, we find
\begin{align*}
\delta_\theta S&=\frac{\a}{2}\int_{M_{10}}e^{-2\phi}\:\tr\:\theta\wedge\Big[e^{2\phi}\:\d_{\nabla^-}(e^{-2\phi}*R^-)-R^-\wedge*H+\Theta^-\wedge e^{2\phi}\:\d(e^{-2\phi}*H)\Big]\\
&-\int_{M_{10}}e^{-2\phi}\delta_\theta B\wedge e^{2\phi}\:\d(e^{-2\phi}*H)\:.
\end{align*}
We can write this as
\begin{equation}
\label{eq:dstheta}
\delta_\theta S=\frac{\a}{2}\int_{M_{10}}e^{-2\phi}\:\tr\:\theta\wedge\textrm{B}_0-\int_{M_{10}}e^{-2\phi}\delta_\theta B\wedge \textrm{H}_0\:,
\end{equation}
since the expression in brackets is proportional to a combination of zeroth order bosonic equations of motion according to \eqref{eq:lemmabos}, while $\textrm{H}_0$ is the zeroth order equation of motion for $H$. It follows from \eqref{eq:dstheta} that the change of connection $\theta$ may be thought of as an $\OO(\a^2)$ field redefinition, as this is precisely how the action gets corrected when we perform such a field redefinition. This is similar to the $\OO(t,\a)$ field redefinitions we described in the previous section. In the same way, it follows that the change of the connection $\theta$ is unphysical. 

Let us next compute the variation of the action \eqref{eq:thetaaction} with respect to the connection $\nabla$, assuming that the first order equations of motion are satisfied. Using the first order equations of motion, in particular
\begin{equation*}
\textrm{H}_0=e^{2\phi}\d(e^{-2\phi}*H)=\OO(\a^2)\:,
\end{equation*}
we find
\begin{align}
\delta_\nabla S\big\vert_{\delta S=\OO(\a^2)}=&\frac{\a}{2}\int_{M_{10}}e^{-2\phi}\:\tr\:\delta\Theta^-\wedge\Big[[\theta,*R^-]+e^{2\phi}\:\d_{\nabla^-}(e^{-2\phi}*\d_{\nabla^-}\theta)\notag\\
&-\d_{\nabla^-}\theta\wedge*H+e^{2\phi}\:\d_{\nabla^-}(e^{-2\phi}*R^-)-R^-\wedge*H\Big]+\OO(\a^3)\:.
\label{eq:deltaSnabla}
\end{align}
Note that any variations depending on $\delta\theta$ drop out of this expression. This is due to \eqref{eq:dstheta} and that $\theta$ is of order $\a$, which implies that variations of the action with respect to $\theta$ are of $\OO(\a^3)$ at the locus of the first order equations of motion. We therefore only need to worry about the $\delta\Theta^-$-part when varying the action with respect to $\nabla$.

Equation \eqref{eq:lemmabos} also guarantees that the expression in \eqref{eq:deltaSnabla} is of $\OO(\a^2)$. The change of the $\OO(\a^2)$ equations of motion depend on what the expression in the brackets is, which again depends on our connection choice. It should be stressed that even though $\theta$ corresponds to a field choice, this does not mean that any field choice will do. {\it We want to choose our fields so that supersymmetry, and in particular the Strominger system, is compatible with the equations of motion.}

Recall that the $\beta$-functions of the $(0,2)$-sigma model correspond to the heterotic supergravity equations of motion. In \cite{Foakes1988335} it was noted that the three-loop $\beta$-function of the gauge connection equal the two-loop $\beta$-function.\footnote{The $n$-loop $\beta$-functions of the sigma models correspond to the $\OO(\a^{n-1})$ supergravity.} That is, the $\beta$-function of the gauge field does not receive corrections at this order, nor should the corresponding supergravity equation of motion. This is consistent with the supergravity point of view \cite{Bergshoeff1989439}. Motivated by this, and guided by the symmetry between $\nabla$ and the gauge connection in the action, it seems natural to choose $\nabla$ so that it satisfies it's own equation of motion
\begin{equation}
\label{eq:nablaeom}
e^{2\phi}\:\d_\nabla(e^{-2\phi}*R)-R\wedge*H=\OO(\a^2),
\end{equation}
at this order. This is exactly the equation one gets when varying the action with respect to $\nabla$, and which is indeed satisfied by the Hull connection at first order. It is easy to see that choosing this connection is in fact equivalent to choosing $\theta$ so that the expression in brackets in \eqref{eq:deltaSnabla} vanishes, modulo higher orders. This again implies that all the first order equations of motion remain the same at $\OO(\a^2)$.

Of course, changing the connection also requires that we change the supersymmetry variations appropriately, in order that the full action remains invariant under supersymmetry transformations at $\OO(\a^2)$. This also relates to how we correct the connection outside of the locus of equations of motion. We will return to this later, when we also consider supersymmetric solutions. We shall see that supersymmetric solutions may be assumed to be solutions of the Strominger system ($\nabla^+\:\eta=0$), without loss of generality. Moreover, they exist if and only if $\nabla$ is an instanton, and in particular \eqref{eq:nablaeom} is satisfied. This is in complete analogy with the gauge connection, as the supersymmetry condition for $A$ is that $F$ remains an instanton at $\OO(\a^2)$ as well.

\subsection{The Hull Connection at $\OO(\a^3)$}
Before we go on to consider the more general connection choices in more detail, let us return to the Hull connection used in \cite{Bergshoeff1989439}. We shall see that choosing this connection severely restricts the allowed supersymmetric solutions. The corrections to the theory in the case of the Hull connection have been worked out in \cite{Bergshoeff1989439} to $\OO(\a^3)$, and we consider this theory to this order.

We look for supersymmetric solutions at $\OO(\a^3)$. As we shall see, insisting upon the Hull connection restricts the allowed supersymmetric solutions. This was also argued in \cite{Ivanov:2009rh} where the first order theory was taken to be exact, which lead to Calabi-Yau solutions as the only consistent solutions.

We will consider the theory up to cubic order in $\a$. At this order, the bosonic action may be given as \cite{Bergshoeff1989439}
\begin{align}
S=&\int_{M_{10}}e^{-2\phi}\Big[*\mathcal{R}-4\vert\d\phi\vert^2+\frac{1}{2}\vert H\vert^2+\frac{\alpha'}{4}(\tr\:\vert F\vert^2-\tr\:\vert R^-\vert^2)\notag\\
&+\frac{3}{2}\a\vert T\vert^2+\frac{\a}{2}T_{IJ}T^{IJ}\Big]+\OO(\a^4),
\label{eq:cubicaction}
\end{align}
where
\begin{align*}
T&=\d H=\frac{\a}{4}\left(\tr\: F\wedge F-\tr\: R^-\wedge R^-\right),\\
T_{IJ}&=\frac{\a}{4}\left(\tr\: F_{IK}{F^K}_J-\tr\: R^-_{IK}{{R^-}^K}_J\right).
\end{align*}
Note again the symmetry between the gauge connection $A$ and the Hull connection $\nabla^-$ in the action above.

The supersymmetry transformations now read
\begin{align}
\label{eq:O3susy1}
\delta\psi_M=&\left(\nabla^+_M+\frac{1}{4}\P_M+\frac{1}{4}\Q_M\right)\epsilon+\OO(\a^4)\:,\\
\label{eq:O3susy2}
\delta\lambda=&-\frac{1}{2\sqrt{2}}\left(\slashed\nabla^{\hbox{\tiny{LC}}}\phi+\frac{1}{12}\H+\frac{3}{12}\P+\frac{3}{12}\Q\right)\epsilon+\OO(\a^4)\:,\\
\delta\chi=&-\frac{1}{2}F_{MN}\slashed F\:\epsilon-\a\bigg[\frac{1}{16}{T_J}^J\slashed F-4T^{LJ}\Gamma^K\Gamma_LF_{JK}\notag\\
\label{eq:O3susy3}
&+\frac{1}{32}\Gamma^{AB}\slashed TF_{AB}\bigg]\epsilon+\OO(\a^3)\:,
\end{align}
where
\begin{align*}
\Q_M=&-24\a^2\:e^{2\phi}\:\nabla^{+I}\left[e^{-2\phi}\:\nabla^+_{[I}(e^{2\phi}\:\nabla^+_{\vert J\vert}e^{-2\phi}\:{T^{J}}_{M]AB]})\right]\Gamma^{AB}\\
&+\a^2\:e^{2\phi}\:\nabla^+_I\bigg[e^{-2\phi}\bigg(\frac{1}{4}{T_J}^J{{R^-_{AB}}^I}_M\Gamma^{AB}\\
&-16T^{LJ}\Gamma^K\Gamma_L{{R^-_{JK}}^I}_M+\frac{1}{8}\Gamma^{AB}\slashed T {{R^-_{AB}}^I}_M\bigg)\bigg]\,,
\end{align*}
and $\Q=\Gamma^M\Q_M$.

Let us now consider supersymmetric solutions. Note that from Proposition \ref{prop:Strom}, by setting $n=2$, we may assume
\begin{equation*}
\nabla^+\:\eta=0\:,
\end{equation*}
without loss of generality. The corresponding supersymmetric solutions are solutions of the Strominger system. It again follows that since the action remains uncorrected at $\OO(\a^2)$, the connection $\nabla$ must still be an instanton at $\OO(\a)$ by Appendix \ref{app:proof}. Insisting on a particular choice of connection may then over constrain the system, as it is not guaranteed that this connection is an instanton. We see how this is true for the particular example of the Hull connection next.

\begin{Theorem}
For compact smooth compactifications, if we insist upon using the Hull connection at $\OO(\a^{n})$, $n\ge2$, we can without loss of generality assume that the first order solution is Calabi-Yau. If we also assume that $\OO(\a)$-corrections are purely geometric, i.e. non-topological, then the second order geometry can be assumed to be Calabi-Yau as well.
\end{Theorem}
\begin{proof}
Let us first consider the theory at $\OO(\a^2)$. From Proposition \ref{prop:Strom}, with $n=2$, we can without loss of generality assume that the geometry solves the Strominger system. We then further need to require the instanton condition
\begin{equation*}
R^-_{pq}\Gamma^{pq}\:\eta=\OO(\a^2)\:,
\end{equation*}
for the Hull connection. At this order in $\a$, this is a nontrivial condition. Indeed, using the identity \eqref{eq:curvatureid}, it follows that
\begin{equation*}
\d H_{mnpq}\Gamma^{mn}\:\eta=\OO(\a^2)\;\;\;\;\Leftrightarrow\;\;\;\;\d H=\OO(\a^2).
\end{equation*} 
It follows that
\begin{equation*}
T=T_{mn}=\OO(\a^2)\:.
\end{equation*}
Using these results, we see that the cubic corrections to the supersymmetry conditions become quartic when requiring the $\OO(\a^2)$ supersymmetry conditions and the equations of motion satisfied. We find that a variation of the cubic corrections to the action at the supersymmetric locus are of $\OO(\a^4)$ as well. 

Arguing this way order by order, it then follows that we need
\begin{equation*}
R^-_{mn}\Gamma^{mn}\:\eta=\OO(\a^3),
\end{equation*}
in order for the equations of motion to be satisfied to cubic order. This further gives the requirement
\begin{equation}
\label{eq;eomO3}
\d H=\OO(\a^3)\:.
\end{equation}
We also have by supersymmetry,
\begin{align*}
\tilde H=e^{-\phi}H=-e^{\phi}\:\d^\dagger(e^{-\phi}*\omega)+\OO(\a^4)=-\d_\phi^\dagger(e^{-\phi}*\omega)+\OO(\a^4)\:,
\end{align*}
where we have defined 
\begin{equation*}
\d_\phi=e^{-\phi}\:\d\: e^{\phi}\:.
\end{equation*}
Note that $\d_\phi^2=0$, with corresponding elliptic Laplacian $\Delta_{\d_\phi}$. It follows that any form $\gamma$ has a Hodge decomposition
\begin{equation*}
\gamma=h_{\d_\phi}+\d_\phi\alpha+\d_\phi^\dagger\beta\:,
\end{equation*}
where $h_{\d_\phi}$ is $\d_\phi$-harmonic. From this we get
\begin{equation*}
\tilde H=-\d_\phi^\dagger(e^{-\phi}*\omega)+\OO(\a^4)=\Delta_{\d_\phi}\kappa+\OO(\a^4)\:,
\end{equation*}
for some three-form $\kappa$. Here
\begin{equation*}
\Delta_{\d_\phi}=\d_\phi\d_\phi^\dagger+\d_\phi^\dagger\d_\phi\:.
\end{equation*}
From \eqref{eq;eomO3}, it follows that $\d_\phi\tilde H=\OO(\a^3)$, which in turn implies
\begin{equation}
\label{eq:kerdelphi}
\Delta_{\d_\phi}\tilde H=\OO(\a^3)\:.
\end{equation}
From this it follows that
\begin{align*}
\vert\vert \tilde H\vert\vert^2=(\tilde H,\Delta_{\d_\phi}\kappa)+\OO(\a^4)=(\Delta_{\d_\phi}\tilde H,\kappa)+\OO(\a^4)=\OO(\a^3)\:,
\end{align*}
from which it follows that
\begin{equation*}
H=\OO(\a^2)\:,
\end{equation*}
where we have excluded fractional powers of $\a$ in the $\a$-expansion of $H$. It follows that the first order geometry is Calabi-Yau.

We can go further if we make a mild assumption about the $\a$-corrections. First note that $\ker(\Delta_{\d_\phi})\cong\ker(\Delta_{\d})$, and $\ker(\Delta_{\d})$ is topological. If we assume that $\a$-corrections are small, and in particular do not change the topology of $X$, it follows that $\vert\ker(\Delta_{\d_\phi})\vert$ does not change under $\a$-corrections. In particular, there are no new zero-modes as $\a\rightarrow 0$. From this it follows that  for $\lambda_i\neq0$ we have $\lambda_i=\OO(1)$. From \eqref{eq:kerdelphi} it follows that
\begin{equation*}
H=\OO(\a^3)\:.
\end{equation*}
This in turn implies that $X$ is a Calabi-Yau, both at first and second order in $\a$.
\end{proof}

\subsection{Choosing Other Connections}
\label{sec:differentcon}
We now consider what happens if a different connection, other than the Hull connection is chosen, that is $\theta\neq0$. We work at $\OO(\a^2)$ for the time being, and leave the cubic and higher order corrections for future work. 

As argued in \cite{Bergshoeff1989439}, the higher order corrections to the supersymmetry transformations come from the failure of $(\Theta,\psi_{IJ})$ to transform as a $SO(9,1)$ Yang-Mills multiplet. Under supersymmetry transformations, we have
\begin{align*}
\delta {\Theta_I}^{JK}=&{(\delta\Theta^-+\delta\theta)_I}^{JK}\\
=&\frac{1}{2}\bar\epsilon\:\Gamma_I\psi^{JK}+\OO(\a)\:,\\
\delta\psi_{IJ}=&-\frac{1}{4}R^+_{IJKL}\Gamma^{KL}\:\epsilon\\
=&-\frac{1}{4}\left({R}_{KLIJ}+\frac{1}{2}\d H_{KLIJ}-\d_{\nabla^-}\theta_{KLIJ}\right)\Gamma^{KL}\:\epsilon\notag\\
=&-\frac{1}{4}{R}_{KLIJ}\Gamma^{KL}\:\epsilon+\OO(\a)\:,
\end{align*}
where \eqref{eq:curvatureid} has been used in the second equality of the expression for $\delta\psi_{IJ}$. Note that without the $\a$-effects, the multiplet transforms as a $SO(9,1)$ Yang-Mills multiplet. This is how the symmetry in the action between the gauge connection and tangent bundle connection arises at $\OO(\a)$.

The $\OO(\a)$ correction to the transformation of ${\Theta_I}^{JK}$ depends on how the correction $\theta$ of the connection is defined in terms of the other fields of the theory. This correction is what makes the action fail to be invariant under supersymmetry transformations. However, this failure of the action to be invariant may be absorbed into an $\OO(\a^2)$-redefinition of the bosonic supersymmetry transformations due to \eqref{eq:lemmabos}, as is done in \cite{Bergshoeff1989439} for the case of the Hull connection. 

The same holds for the $\OO(\a)$ correction to $\delta\psi_{IJ}$,
\begin{equation*}
\delta_{\a}(\delta\psi_{IJ})=-\frac{1}{8}\Big(\d H_{KLIJ}-2\d_{\nabla^-}\theta_{KLIJ}\Big)\Gamma^{KL}\:\epsilon\:.
\end{equation*}
This can be absorbed into a redefinition of the supersymmetry transformations of the fermions due to \eqref{eq:lemmaferm}. For the more general connection choice, it turns out that the correction we need only requires a change of the three-form $P$,
\begin{align}
\label{eq:redefP}
P_{MAB}\rightarrow\tilde P_{MAB}=-\a6e^{2\phi}\:\nabla^{+L}e^{-2\phi}\Big(\d H_{LMAB}-2(\d_{\nabla^-}\theta)_{ABLM}\Big)\:,
\end{align}
but otherwise the transformations \eqref{eq:spinorsusy1}-\eqref{eq:spinorsusy3} remain the same. Note also that as the deformation of the connection can again be viewed as an $\OO(\a^2)$ field redefinition, the new supersymmetry algebra is again closed.

We now compactify our theory on a complex three-fold $X$. By the argument given in Proposition \ref{prop:Strom}, with $n=2$, we can assume without loss of generality that
\begin{equation*}
\nabla^+_m\:\eta=\OO(\a^3)\:.
\end{equation*}
By the rewriting of the bosonic action \eqref{eq:6daction}, which we stress holds true at $\OO(\a^2)$, we find that for the equations of motion to hold we need $R=R^-+\d_{\nabla^-}\theta+\OO(\a^2)$ to satisfy the instanton condition, 
\begin{equation}
\label{eq:O2instanton}
R_{mn}\Gamma^{mn}\:\eta=\OO(\a^2)\:.
\end{equation}
Note the similarity between this condition and the supersymmetry condition for the gauge field \eqref{eq:spinorsusy3}.

Supersymmetry now also requires
\begin{equation*}
\tilde P_{mAB}\Gamma^{AB}\:\eta=\OO(\a^3)\;,
\end{equation*}
by \eqref{eq:redefP} and \eqref{eq:spinorsusy1}. Here $A,\:B$ denote flat indices on $X$.  This equation is however trivial, once we know that $R$ is an instanton. Indeed, we have
\begin{align*}
\tilde P_{mAB}\Gamma^{AB}\:\eta=&-6\a\: e^{2\phi}\:\nabla^{+n}e^{-2\phi}\Big(\d H_{nmAB}-2(\d_{\nabla^-}\theta)_{ABnm}\Big)\Gamma^{AB}\:\eta\\
=&12\a\:e^{2\phi}\:\nabla^{+n}e^{-2\phi}\Big(R_{ABnm}\Gamma^{AB}\:\eta\Big)+\OO(\a^3)\\
=&\OO(\a^3)\;,
\end{align*}
where we used \eqref{eq:curvatureid} in the second equality.

It should also be mentioned that the instanton connection solves the $\nabla$-equation of motion \eqref{eq:nablaeom}, as shown in \cite{Gauntlett:2002sc, Gillard:2003jh}. Indeed in dimension six, by the supersymmetry conditions, it follows that
\begin{equation*}
e^{2\phi}\:\d_\nabla(e^{-2\phi}*R)-R\wedge*H=e^{2\phi}\:\d_\nabla\: e^{-2\phi}(*R+R\wedge\omega)\:.
\end{equation*}
As $R$ is both of type $(1,1)$, and primitive, we have the identity $*R=-\omega\wedge R$. It follows that the instanton connection satisfies the $\nabla$-equation of motion, and the first order equations of motion do not receive corrections. 

We have thus gone through the proof of the statements in Theorem \ref{tm:inst}. Next, we want to consider their interpretation and give a discussion of the results. In doing so we also give our reasons for proposing Conjecture \ref{con:higherorder}.

\section{Discussion}
\subsection{Summary of Results}
\label{sec:disc}
In the first order theory, we saw that the connection $\nabla=\nabla^-+t\theta$, where $\theta$ depends on the fields of the theory in some way, should satisfy the instanton condition whenever the solution is supersymmetric of Strominger system type. As shown in e.g. \cite{delaOssa:2014cia}, this condition has an infinitesimal moduli space of the form
\begin{equation}
\label{eq:Mdnabla}
T\M_{\bp_{\nabla^-}}\cong H^{(0,1)}_{\bp_{\nabla^-}}(X,\End(TX))\:,
\end{equation}
where the tangent space is taken {\it at the Hull connection}. At first order, the requirement that the full supergravity action should be invariant under the usual supersymmetry transformations reduces the choice to the Hull connection. Hence, the $t$-deformed theory requires changes to the supersymmetry transformations, and we found what these where. We also saw that the allowed deformation space of connections, for which supersymmetric solutions of the Strominger system exist, was given by \eqref{eq:Mdnabla}. Supersymmetric solutions could also be assumed to be solutions of the Strominger system by Proposition \ref{prop:Strom}. Moreover, by the lemma of Bergshoeff and de Roo \cite{Bergshoeff1989439}, these deformations correspond to infinitesimal $\OO(\a)$ field redefinitions.

Returning to the usual form of the first order supergravity, we saw that at second order the theory can again be corrected appropriately for any $\OO(\a)$-change $\theta$ of the Hull connection $\nabla^-$, corresponding to $\OO(\a^2)$ field redefinitions. Supersymmetric solutions could again be assumed to be solutions of the Strominger system, and the equations of motion are compatible with supersymmetry if and only if $\nabla=\nabla^-+\theta$ satisfies the instanton condition again.

\subsection{Higher orders}
Let us now take a moment to discuss higher orders in $\a$. Note that the condition we find for compatibility between supersymmetry and equations of motion, \eqref{eq:O2instanton}, is exactly the supersymmetry condition we would get from this ``connection sector" if $\nabla$ was part of a dynamical superfield, very much analogous to the gauge sector. Indeed, the fact that $(\nabla^-,\psi_{IJ})$ transforms as an $SO(9,1)$-Yang-Mills multiplet to $\OO(\a)$ is what motivated the construction of the action of \cite{Bergshoeff1989439} in the first place. From the discussion above, it appears that supersymmetric solutions behave as if this where the case, at least for compact solutions including $\OO(\a^2)$. The question then arises what happens at $\OO(\a^3)$ and higher?

It should first be noted that at higher orders, the form of the supergravity action is no longer unique, and undetermined $(\textrm{curvature})^4$-terms appear \cite{Bergshoeff1989439}. The form of these terms may however be determined through other means such as string amplitude calculations \cite{Gross198741, Cai1987279}, which was also used in \cite{Bergshoeff1989439}, and these terms indeed preserve the symmetry between the Lorentz and Yang-Mills sectors.

With this, it therefore seems natural to conjecture that the above structure also survives to higher orders. That is, the natural connection $\nabla$ used to calculate the curvatures should be chosen so that it satisfies an equation of motion similar to that of $A$, whenever the other equations of motion are satisfied. Moreover, for supersymmetric solutions, $\nabla$ should satisfy a supersymmetry condition similar to that of $A$. We also conjecture that, as seen to order $\OO(\a)$, the moduli of this ``supersymmetry condition" are equivalent to field redefinitions, and therefore do not correspond to physical lower energy fields in any sense.  

\subsection{Future directions}
Having reviewed our results, and discussed higher orders in $\a$, there are a few unanswered questions which we would like to look into in the future. Firstly, it would be interesting to check the proposed conjecture to the next order in $\a$. This should not be very difficult, as the cubic theory was laid out in general in \cite{Bergshoeff1989439}, and we only have to repeat their analysis using a more general connection. It should be noted that  at this order, the supersymmetry condition for the gauge field does receive corrections \eqref{eq:O3susy3}, and we expect this to be true for the tangent bundle connection as well. It would be interesting to check whether to cubic or higher order in $\alpha‚Äô$, supersymmetric solutions still satisfy the Strominger system.  As noted in \cite{delaOssa:2014cia}, these solutions can be recast in terms of a holomorphic structure $\bar D$ on a generalised bundle
\begin{equation*}
\Q=T^*X\oplus\End(TX)\oplus\End(V)\oplus TX\:,
\end{equation*}
and it would be interesting to see if this structure survives beyond second order as well. This might be expected to be the case, as it is suggested by the authors in \cite{Coimbra:2014qaa}, where it is argued that the generalised geometric structure introduced on ${\cal Q}$ survives to higher orders.

Next, it would be interesting to return to the first order theory, and consider higher order deformations of the Hull connection. Indeed, in section \ref{sec:FirstOrder} we only considered infinitesimal deformations away from the Hull connection of this theory. That is, we considered the tangent space of the moduli space of connections {\it at the Hull connection}, which we saw corresponded to infinitesimal $\OO(\a)$ field redefinitions. It would be interesting to perform higher order deformations of the connection, i.e. deformations of $\OO(t^2)$ and above, and to see how this relates to obstructions of the corresponding deformation theory. Moreover, do such ``finite" deformations still correspond to field redefinitions? 

It would also be interesting to consider our findings in relation to the sigma model. In terms of the first order sigma model, it was pointed out in \cite{Sen1986289} that changing the connection $\nabla$ corresponds to $\OO(\a)$ field redefinitions, consistent with the findings of the present paper. Requiring world-sheet conformal invariance, i.e. the ten-dimensional equations of motion, in addition to space-time supersymmetry, puts conditions on the connection. As we have seen, and as was first noted in \cite{Hull1986357}, it is sufficient to use the Hull connection at first order. This connection was also necessary modulo field redefinitions. We found that the allowed field redefinitions correspond to the moduli space \eqref{eq:Mdnabla}, and it would be interesting to see if this moduli space can be retrieved from the sigma model point of view as well.

At next order, we found that the Hull connection was not a good field choice, provided one wants supersymmetric solutions to the Strominger system. Still, we found the necessary and sufficient condition for compatibility was that $\nabla$ should satisfy the instanton condition. Moreover, $\nabla$ is related to the Hull connection by a corresponding $\OO(\a^2)$ field redefinition. But for supersymmetric solutions of the Strominger system, the Hull connection lead to too stringent constraints on the geometry, leading us back to Calabi-Yau. It would be interesting to investigate this further from a sigma-model point of view. In particular, it would be interesting to see what the more ``natural field choices", i.e. connections satisfying the instanton condition, look like in this picture.

\section*{Acknowledgments}
We would like to thank Philip Candelas, James Gray, Jock McOrist, Stefan Groot Nibbelink, Savdeep Sethi, Nick Halmagyi, Dan Israel, Atish Dabholkar for useful conversations during the completion of this work. 
ES would like to thank LPTHE for hospitality while the final stages of this work was carried through.  
XD's research is supported in part by the EPSRC grant BKRWDM00. 
ES's research was supported by the Clarendon Scholarship of OUP and a Balliol College Dervorguilla Scholarship. Part of this work was also made in the ILP LABEX (under reference ANR-10-LABX-63), and was supported by French state funds managed by the ANR within the Investissements d’Avenir programme under reference ANR-11-IDEX-0004-02.

\appendix

\section{Compactification and $SU(3)$-structures}
\label{app:comp}
In this Appendix, we review compactifications of heterotic supergravity on spaces of type $M_4\times X$ where $M_4$ is four-dimensional Minkowski space, and $X$ is compact. We focus on the topological aspects of the compactification, leaving the discussion of supersymmetry to the main part of the paper.

The ten-dimensional geometry is now postulated to have the form of a direct product,
\begin{equation*}
M_{10}=M_4\times X,
\end{equation*}
where $M_4$ is four-dimensional space-time, and $X$ is a compact six-dimensional internal space. We will use small roman indices $(m,n,p,..)$ to denote indices on $X$, and greek indices to denote indices on $M_4$.

The spinor $\epsilon$ decomposes as
\begin{equation}
\label{eq:decompspinor}
\epsilon=\rho\otimes\eta,
\end{equation}
where $\rho$ is a four-dimensional space-time spinor, while $\eta$ is a spinor on $X$. This spinor defines an $SU(3)$-structure with two- and three-forms $\omega$ and $\Psi$ on $X$ given by
\begin{align*}
\omega_{mn}&=-i\eta^\dagger\gamma_{mn}\eta\:,\\
\Psi_{mnp}&=\eta^T\gamma_{mnp}\eta\:,
\end{align*}
where ${\gamma^m}$ are six-dimensional gamma-matrices. These forms satisfy the usual $SU(3)$-structure identities
\begin{equation*}
\omega\wedge\Psi=0,\;\;\;\;\frac{i}{\vert\vert\Psi\vert\vert^2}\Psi\wedge\bar\Psi=\frac{1}{6}\omega\wedge\omega\wedge\omega.
\end{equation*}
The three-form $\Psi$ may also be used to define an almost complex structure $J$ on $X$ by
\begin{equation}
\label{eq:complexstr}
{J_m}^n=\frac{{I_m}^n}{\sqrt{-\frac{1}{6}\tr I^2}},
\end{equation}
where the endomorphism $I$ is given by
\begin{equation*}
{I_m}^n=(\textrm{Re}\Psi)_{mpq}(\textrm{Re}\Psi)_{rst}\epsilon^{npqrst}.
\end{equation*}
The normalization in \eqref{eq:complexstr} is needed so that $J^2=-1$. Note also that the complex structure $J$ is independent of rescalings of $\Psi$.

A general $SU(3)$-structure is parameterised by five torsion classes, \newline $(W_0, W_1^\omega,W_1^\Psi,W_2,W_3)$ where \cite{0444.53032, 1024.53018, LopesCardoso:2002hd, Gauntlett:2003cy}
\begin{align*}
\d\omega&=-\frac{12}{\vert\vert\Psi\vert\vert^2}\textrm{Im}(W_0\bar\Psi)+W_1^\omega\wedge\omega+W_3\\
\d\Psi&=W_0\:\omega\wedge\omega+W_2\wedge\omega+\bar W_1^\Psi\wedge\Psi\:.
\end{align*}
Here $W_0$ is a complex function, $W_2$ is a primitive $(1,1)$-form, $W_3$ is real and primitive and of type $(1,2)+(2,1)$. Also, $W_1^\omega$ is a real one-form, while $W_1^\Psi$ is a $(1,0)$-form. These are known as the Lee-forms of $\omega$ and $\Psi$ respectively, and they are given by
\begin{align*}
W_1^\omega&=\frac{1}{2}\omega\lrcorner\d\omega\:,\\
W_1^\Psi&=\frac{1}{\vert\vert\Psi\vert\vert^2}\Psi\lrcorner\d\bar\Psi\:.
\end{align*}
It should be noted that $W_2=W_0=0$ is equivalent to the vanishing of the Nijenhaus tensor, and therefore equivalent to $X$ being complex.

\section{Proof of Instanton Condition}
\label{app:proof}
In this Appendix, we repeat the proof of \cite{delaOssa:2014cia} showing that supersymmetric solutions of the Strominger system, and the equations of motion are compatible if and only if $\nabla$ is of instanton type. We consider the theory including $\OO(\a^2)$ terms. We note that this is a special case of a more general proof, which appeared in \cite{Held:2010az}.

Recall first that the second order bosonic action is the same as the first order action \cite{Bergshoeff1989439}. According to \cite{Cardoso:2003af}, the six-dimensional part of the action \eqref{eq:action} may be written in terms of $SU(3)$-structure forms as
\begin{align}
\label{eq:6daction}
S_6=&\frac{1}{2}\int_{X}e^{-2\phi}\Big[-4\vert\d\phi-W_1^\omega\vert^2+\omega\wedge\omega\wedge\mathcal{\hat R}+\vert H-e^{2\phi}*\d(e^{-2\phi}\omega)\vert^2\Big]\notag \\
&-\frac{1}{4}\int\d ^6y\sqrt{g_6}{N_{mn}}^pg^{mq}g^{nr}g_{ps}{N_{nq}}^s\notag\\
&-\frac{\alpha'}{2}\int_{X}e^{-2\phi}\Big[\tr\vert F^{(2,0)}\vert^2+\tr\vert F^{(0,2)}\vert^2+\frac{1}{4}\tr\vert F_{mn}\omega^{mn}\vert^2\Big]\notag\\
&+\frac{\alpha'}{2}\int_{X}e^{-2\phi}\Big[(\tr\vert{R}^{(2,0)}\vert^2+\tr\vert {R}^{(0,2)}\vert^2+\frac{1}{4}\tr\vert R_{mn}\omega^{mn}\vert^2\Big]+\OO(\a^3)\:,
\end{align}
where the Bianchi Identity has been applied. $\mathcal{\hat R}$ is now the Ricci-form of the unique connection $\hat\nabla$ with totally antisymmetric torsion, for which the complex structure is parallel. For supersymmetric solutions of the Strominger system, the connection $\nabla^+$, see equation \eqref{eq:connsymb}, coincides with $\hat\nabla$, which is known as the Bismut connection in the mathematics literature. The Ricci-form is
\begin{equation*}
\mathcal{\hat R}=\frac{1}{4}\hat R_{pqmn}\omega^{mn}\d x^p\wedge\d x^q,
\end{equation*}
while ${N_{mn}}^p$ is the Nijenhaus tensor for this almost complex structure. Note that
\begin{equation*}
\mathcal{\hat R}=0
\end{equation*}
is an integrability condition for supersymmetry.

Performing a variation of the action at the supersymmetric locus, we find that most of the terms vanish. The only surviving terms are
\begin{align}
\delta S_6&=\frac{1}{2}\int_{X}e^{-2\phi}\omega\wedge\omega\wedge\delta\mathcal{\hat R}\notag\\
\label{eq:6dvary}
&+\frac{\alpha'}{2}\delta\int_{X}e^{-2\phi}\Big[(\tr\vert{R}^{(2,0)}\vert^2+\tr\vert {R}^{(0,2)}\vert^2+\frac{1}{4}\tr\vert R_{mn}\omega^{mn}\vert^2\Big]+\OO(\a^3).
\end{align}
In \cite{Cardoso:2003af} it is shown that $\delta\mathcal{\hat R}$ is exact, and therefore the first term vanishes using supersymmetry by an integration by parts. If the equations of motion are to be satisfied to the order we work at, we therefore find
\begin{equation*}
R_{mn}\Gamma^{mn}\:\eta=\OO(\a^2)\:,
\end{equation*}
which is equivalent to the instanton condition. Note the reduction in orders of $\a$ due to the factor of $\a$ in front of the curvature terms in the action.

\section{The Hull connection}
\label{app:Hull}
For completeness, we also repeat the argument of \cite{Martelli:2010jx} that the Hull connection does indeed satisfy the instanton condition for the $\OO(\a)$-theory, whenever we have supersymmetry. We work in ten dimensions in this Appendix.

It is easy to show that
\begin{equation}
\label{eq:curvatureid}
R^+_{MNPQ}-R^-_{PQMN}=\frac{1}{2}\d H_{MNPQ}.
\end{equation}
At zeroth order we get
\begin{equation*}
R^+_{MNPQ}=R^-_{PQMN}+\OO(\a),
\end{equation*}
by the Bianchi Identity. Contracting both sides with $\Gamma^{PQ}\:\eta$, and using 
\begin{equation*}
R^+_{MNPQ}\Gamma^{PQ}\:\epsilon=\OO(\a^2)
\end{equation*}
at the supersymmetric locus, we find
\begin{equation*}
R^-_{PQ}\Gamma^{PQ}\:\epsilon=\OO(\a)\:,
\end{equation*}
as required.

\providecommand{\href}[2]{#2}\begingroup\raggedright\endgroup

\end{document}